\documentclass[journal,draftclsnofoot,onecolumn,11pt]{IEEEtran}
\usepackage{latexsym,amssymb,amsmath,graphicx,epsfig,cite,bbm,float,subfig}
\usepackage{ifpdf}
\def\ninept{\def\baselinestretch{1.5}}
\ninept

\def\ninept{\def\baselinestretch{1.5}}
\ninept

\newtheorem{definition}{{\bf \em Definition}}[section]
\newtheorem{remark}{{\bf \em Remark}}[section]
\newtheorem{theorem}{{\bf \em Theorem}}[section]

\newtheorem{lemma}{{\bf \em Lemma}}[section]
\newtheorem{proposition}{{\bf \em Proposition}}[section]

\newcommand{\eqdef} {\mbox{$\:\stackrel{\triangle}{=}\:$}}

\newcommand{\bx}{{\mathbf{x}}}
\newcommand{\bX}{{\mathbf{X}}}
\newcommand{\by}{{\mathbf{y}}}

\newcommand{\bu}{{\mathbf{u}}}

\newcommand{\bY}{{\mathbf{Y}}}
\newcommand{\bU}{{\mathbf{U}}}

\newcommand{\tpi}{{\tilde{\pi}}}

\newcommand{\ep}{{\epsilon}}

\newcommand{\cE}{{\mathcal{E}}}

\newcommand{\bbR}{{\mathbb{R}}}
\newcommand{\cC}{{\mathcal{C}}}

\newcommand{\cR}{{\mathcal{R}}}

\newcommand{\cW}{{\mathcal{W}}}
\newcommand{\hW}{{\hat{W}}}
\newcommand{\cU}{{\mathcal{U}}}

\newcommand{\cX}{{\mathcal{X}}}
\newcommand{\cY}{{\mathcal{Y}}}

\newcommand{\bpi}{{\bar{\pi}}}

\newcommand{\be}{\begin{equation}}
\newcommand{\ee}{\end{equation}}
\newcommand{\bea}{\begin{eqnarray}}
\newcommand{\eea}{\end{eqnarray}}
\newcommand{\bean}{\begin{eqnarray*}}
\newcommand{\eean}{\end{eqnarray*}}
\newcommand{\ben}{\begin{enumerate}}
\newcommand{\een}{\end{enumerate}}

\newcommand{\qed}{\hspace*{\fill}%
    \vbox{\hrule\hbox{\vrule\squarebox{.667em}\vrule}\hrule}\smallskip}
    \def\squarebox#1{\hbox to #1{\hfill\vbox to #1{\vfill}}}

\begin{document}

\title{Information Theoretic Analysis of the Fundamental Limits of Content Identification}

\author{Sait Tun\c{c}, Y\"{u}cel Altu\u{g}, Suleyman S. Kozat, {\em Senior Member}, IEEE, M. Kivanc Mihcak, {\em Member}, IEEE \thanks{S. Tun\c{c} is with the Booth School of Business at University of Chicago, Chicago, IL, Y. Altu\u{g} is with the Electrical and Computer Engineering Department at Cornell University, Ithaca, NY, S. S. Kozat is with the Electrical and Electronic Engineering Department at Bilkent University, Ankara, Turkey, M. K. M{\i}h\c{c}ak is with the Coordinated Science Laboratory at University of Illinois, Urbana-Champaign, Urbana, IL.}}

\maketitle

\begin{abstract}

We investigate the content identification problem from an information
theoretic perspective and derive its fundamental limits.  Here, a
rights-holder company desires to keep track of illegal uses of its
commercial content, by utilizing resources of a security company,
while securing the privacy of its content.  Due to privacy issues, the
rights-holder company only reveals certain hash values of the original
content to the security company. We view the commercial content of the
rights-holder company as the codebook of an encoder and the hash
values of the content (made available to the security company) as the
codebook of a decoder, i.e., the corresponding codebooks of the
encoder and the decoder are not the same.  Hence, the
content identification is modelled as a communication problem
using asymmetric codebooks by an encoder and a decoder. We further
address ``the privacy issue'' in the content identification by adding
``security'' constraints to the communication setup to prevent
estimation of the encoder codewords given the decoder codewords. By
this modeling, the proposed problem of reliable communication with
asymmetric codebooks with security constraints provides the
fundamental limits of the content identification problem. To this end,
we introduce an information capacity and prove that this capacity is
equal to the operation capacity of the system under i.i.d.  encoder
codewords providing the fundamental limits for content
identification. As a well known and widely studied framework, we
evaluate the capacity for a binary symmetric channel and provide 
closed form expressions.

\end{abstract}

\begin{center} \vspace{-0.1in}
\bfseries EDICS Category:  INF-CONF, ADP-PMOD, 	ADP-PPRO, MMH-BENM-PER \vspace{-0.1in}
\end{center}

\section{Introduction}
\label{sec:intro}

In recent years, the ``content identification'' problem has attracted
a growing interest from the signal processing community due to its
potential usage as a filtering technique for file and multimedia
sharing \cite{Moulin10, venka00, kozat04, mihcak01, ward11, varna11, 
jang10, beekhof10, haitsma02, joly05}. Currently,
several video sharing sites, including Youtube, Google Video and
Dailymotion, planted content identification technology in order to
allow copyright holders to identify and disable illegally uploaded
versions of their content in real-time. In these content identification
applications, comparison of the whole file is naturally inefficient
and in some cases impossible due to computational complexity. However,
instead of matching the whole content, using robust hashing methods
enables real-time content identification possible, where short fingerprints
extracted from the content are matched 
\cite{Moulin10, ward11, lee08, baluja07}. To this end, this paper particularly focuses on the
information-theoretic analysis of the fundamental limits of content
identification by modeling the content identification as a communication problem
using asymmetric codebooks under certain security constraints.

In the most generic content identification framework, one seeks to
find an efficient method to perform an ``anti-piracy search'' via side
information at the receiver side. In this framework, a ``rights-holder
company'', i.e., a company that possesses a commercially valuable
signal such as a video, an audio or a document, seeks to identify illegal
uses of its commercial content. However, due to the lack of necessary
infrastructure, the rights-holder company is required to employ 
resources of another company, say a ``security company'', in order to
perform the content identification. As a consequence, the rights-holder
company forms a collaboration with the security company to carry out
the illegal content search. However due to obvious privacy issues,
instead of revealing the whole private content, the rights-holder
company only provides certain hash values, extracted from the original
content, to the security company. The security company needs to
perform content identification only with the help of these hash values,
i.e., side information, derived from the private content revealed by the
rights-holder company.

In this paper, we adapt an information-theoretic approach to the
content identification problem and derive the fundamental limits by
modeling content identification as a communications problem. In this
framework, we view the content owned by the rights-holder company as
the codebook of an encoder. The illegal uploading of the content,
which possibly includes noise, corresponds to the message transmission
stage of the noisy communication channel. The hash values of the
content made available to the security company, i.e., the side
information, correspond to the codebook of the decoder. Hence, the
content identification problem is modelled as a communication problem
in which the encoder and the decoder are communicating with each
other, i.e., the encoder seeks to send a message to the decoder, while
the encoder would like to maintain a reliable communication with the
encoder. Note that due to security requirements, the encoder does not
reveal its codebook to the decoder. Instead, the encoder shares a
perturbed version of its codebook with the decoder. Hence the
corresponding codebooks of the encoder and the decoder are not the
same, i.e., there is an asymmetry between these codebooks. Therefore
the content identification is modelled as a communications problem
with a noisy channel, where there is an asymmetry between the
codebooks of the encoder and the decoder. We further address “the
privacy issue” in the content identification problem by adding
security constraints to the communication setup to prevent the
estimation of the encoder codewords given the decoder codewords. By
this modeling, the proposed problem of reliable communication with
asymmetric codebooks with security constraints provides the
fundamental limits of the content identification problem. Under this
framework, we derive and characterize the maximum achievable rate of
reliable communication. We further evaluate our results for a binary
symmetric case, where the encoder codebook is binary, the perturbation
between the codebooks of the encoder and the decoder is a binary
symmetric distribution and the communication channel is a binary
symmetric channel.

In particular, to provide the fundamental limits on content
identification with security constraints, we study a point-to-point
communication problem, where the communication channel employs
asymmetric codebooks, i.e., the encoder and the decoder codebooks are
not the same. We concentrate on the case where the decoder's codebook
is a perturbed version of the encoder codebook and optimize the
reliable communication rate over the joint statistical distribution of
the codebooks of the encoder and decoder. Thus, we consider the
statistical characterizations of both the encoder codebook and the
perturbation and carry out optimization in the general case. In this
sense, we generalize the original point-to-point communication setup
proposed by Shannon \cite{shan48} by introducing reliable communication using
asymmetric codebooks.

Problems related to the capacity of communication channels with side
information is heavily investigated in the information theory
literature \cite{shan58,wyn75,wyn76,costa83}. However, the introduced
asymmetric codebook nature of our problem with certain security
constraints significantly differentiates the content identification
setup from that of the generic side information related problems. Note
that the generic side information related problems studied in the
literature usually includes a common codebook shared by both the
encoder and the decoder \cite{shan58,wyn75}. On the contrary, here,
due to the nature of content identification application, the codebooks are
asymmetric. Furthermore in the generic communications problems with side information,
either the transmitter or the receiver has access to side information,
e.g., information about certain system parameters or the noise, which
is not available to the other. However, in our setup, the system
parameters, which correspond to the statistical characterization of
the elements of the system, are available for both the encoder and the
decoder. Moreover, while a similar communication problem using
asymmetric codebooks is studied in \cite{Alt08}, the statistical
characterization of the perturbation between the codebooks of the
encoder and the decoder is assumed to be fixed, which reduces the
generality of this setup. Hence the privacy of the valuable content is
not guaranteed while the rate of reliable communication is optimized
in \cite{Alt08}, unlike this paper. This paper derives and
characterizes the maximum achievable rate of reliable communication
while maintaining  ``security'' of the valuable content after
introducing  ``security'' conditions from an estimation theoretic
perspective.

In this paper, we first characterize the fundamental limits of the described
content identification application by deriving the maximum rate of
error-free information transfer of the prescribed communication
setup. We consider the case where the codewords of encoder are drawn
identically and independently from a discrete and finite set. We
further assume that the communication channel between the encoder and
the decoder and the statistical characterization of the perturbation
between the encoder and the decoder codebooks, i.e., the signal
hashing in content identification, are memoryless. We then provide the
maximum achievable rate of reliable communication, which is shown to be
the maximum of mutual information between the codeword of the decoder
and the output of the channel, where the maximization carried over a
set of joint distribution of the encoder's codeword and the decoder's
codeword satisfying certain security constraints. Furthermore, we also
evaluate the capacity for binary symmetric setup, i.e., the alphabet
where the encoder codewords are drown is binary, the perturbation
between the codebooks is a binary symmetric distribution and the
communication channel is a binary symmetric channel.

We begin with the notation and the problem description in
Section~\ref{sec:notation-problem-statement}.  In
Section~\ref{sec:general-results}, we derive and characterize the
corresponding maximum achievable rate of reliable
communications. Then, in Section \ref{ssec:achievability} and Section
\ref{ssec:converse}, we provide proofs for the forward and the
converse statements, respectively. In Section~\ref{sec:binary}, we
analyze the capacity of the binary symmetric case and provide a closed
form expression. The paper concludes with discussions in
Section~\ref{sec:conclusions}.

\section{Notation and Problem Setup}
\label{sec:notation-problem-statement}

\subsection{Notation}
\label{ssec:notation}

Boldface letters and regular letters with subscripts denote vectors and
individual elements of vectors, respectively. Furthermore, capital letters
and lowercase letters denote random variables and individual
realizations of the corresponding random variable, respectively. The vector
$\left[ a_1, a_2, \ldots , a_n \right]^T$ is denoted
by $\mathbf{a}^n$. The abbreviations ``i.i.d.'', ``p.m.f.'',
and ``w.l.o.g.'' are shorthands for the terms
``independent identically distributed'', ``probability mass
function'',  and ``without loss
of generality'', respectively.
The entropy function of a discrete random variable $X$ is denoted by
$H \left( X \right) = - \sum_{x \in \cX} p \left( x \right) \log p \left( x \right)$
where $X$ is defined on the alphabet $\cX$ with the corresponding p.m.f.
$p \left( x \right)$\footnote{Unless otherwise stated, all the logarithms
are base-$2$.}. Similarly, $H \left( X , Y \right)$, $H \left( X | Y \right)$, $I \left( X ; Y \right)$
denote the joint entropy of $X$ and $Y$, conditional entropy of $X$ given $Y$,
and the mutual information between $X$ and $Y$ for discrete random variables $X$ and $Y$, respectively.

\subsection{Problem Setup and Relevant Definitions}
\label{ssec:problem-statement}

We consider a content identification problem, (1) where a
rights-holder company desires to keep track of the illegal uses of its
commercial content (2) by utilizing the resources of a security
company (3) while securing the privacy of the content. The illegal
uses of the content can be broadcasting, uploading or publishing the
original content or a slightly disturbed version of it without proper
consent of the rights-holder company.  Due to privacy issues, the
rights-holder company only reveals certain hash values of the original
content to the security company. In this sense, the rights-holder
company wants to ensure the privacy of the content while providing
sufficient side information, i.e., hash values (which makes the
anti-piracy search feasible), on the original content to the security
company.

We view the commercial content of the rights-holder company as the
codebook of an encoder and the hash values of the content (made
available to the security company) as the codebook of a decoder. Hence
the corresponding codebooks of the encoder and the decoder are not the
same, i.e., there is an asymmetry between these codebooks due to the
described nature of the content identification problem. Furthermore we
model the illegal uploading or broadcasting the content as the message
transmission phase of the communication framework and the
identification of the illegal content as the decoding the output of
the corresponding communication setup. Therefore the content
identification problem is modelled by a communication framework with a
noisy channel, where there is an asymmetry between the codebooks of
the encoder and the decoder. We further address ``the privacy issue''
in the content identification problem by adding security constraints
to the communication setup to prevent the estimation of the encoder
codewords given the decoder codewords. By this modeling, the proposed
problem of reliable communication with asymmetric codebooks with
security constraints constitutes the fundamental limits of the content
identification problem. To this end, we first introduce the
corresponding discrete memoryless communication channel setup and
define the related error events. We then rigorously characterize the
security constraints of this communication setup to address the
privacy issue in content identification by introducing security
related definitions from an estimation theoretic perspective.

We first provide the necessary channel code and related error events
to quantify the fundamental rates of reliable communications with
asymmetric codebooks, which in turn reveals the fundamental limits of the content
identification problem. A broad definition of {\em asymmetric channel
  codes}, which constitutes the fundamental part of reliable
communications with asymmetric codebooks, will be provided.  Such a
communication system, which is depicted in Fig.~\ref{fig:setup},
consists of two components: a \emph{discrete-memoryless communication
  channel} (DMCC) denoted by $\left( \cX, p\left( y| x \right), \cY
\right)$ (with single letter input alphabet $\cX$, single letter
output alphabet $\cY$, single letter transition probability $p \left(
y | x \right)$, cf. \cite{cover}, p.~193) and a
\emph{$\left(2^{nR},n\right)$ asymmetric channel code}
\footnote{Throughout the paper, for the sake of convenience, we assume that $2^{nR} \in {\mathbb Z}^+$ for all
$R \in {\mathbb R}^+ \cup \left\{ 0 \right\}$ and for any $n \in {\mathbb Z}^+$.} (cf. Def.~\ref{def:ACC}).

\begin{figure}
\centering \epsfxsize 6in
 \epsfbox{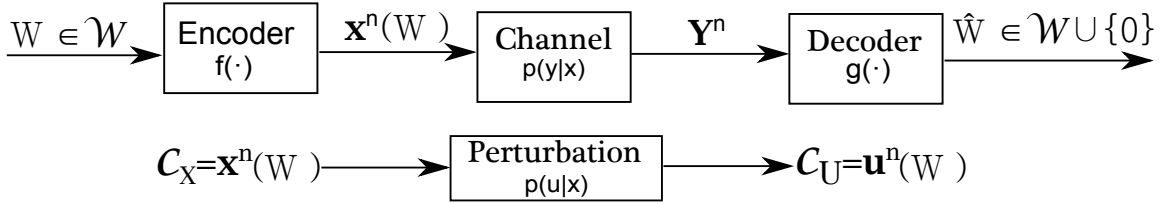}
 \caption{\small Schematic diagram of a discrete memoryless channel with asymmetric codebooks.}
\label{fig:setup}
\end{figure}

\begin{definition}
Given {\em discrete finite} alphabets $\cX$, $\cU$, $\cY$,
a \emph{$\left( 2^{nR},n\right)$ asymmetric channel code (ACC)},
denoted by $\left( \cW , f^n , h^n , g^n \right)$ consists of:  a message set, $\cW \eqdef \left\{ 1, \ldots,
2^{nR}\right\}$; a deterministic encoder codebook generator function, $f^n \, : \, \cW \rightarrow \cX^n$,
where $f^n \left( w \right) = \bx^n \left( w \right)$ for all $w \in \cW$; a deterministic decoder codebook generator function,
$h^n \, : \, \cW \rightarrow \cU^n$, where $h^n \left( w \right) = \bu^n \left( w \right)$ for all $w \in \cW$; and a
deterministic decoding function, $g^n \, : \, \cY^n \rightarrow \cW \cup \{ 0 \}$, which assigns a decision
(denoted by $\hW$) to every received sequence $\by^n \in \cY^n$, where the decision of ``null'' is denoted by $0$.
\label{def:ACC}
\end{definition}

\begin{remark}
In the considered setup, $W \in \cW$ represents the message, which is uniformly distributed over the discrete finite
set $\cW$. Following the standard information-theoretic notation, we use the shorthands
$\bX^n := \bx^n \left( W \right)$, $\bU^n := \bu^n \left( W \right)$.  Here, $\bY^n \in \cY^n$ denotes the output
 of the DMCC $\left( \cX, p\left( y| x \right), \cY \right)$  when the input is $\bX^n$. Note that for the related
content identification problem, the encoder codebook $\bX^n$ corresponds to the commercial content and transmission of
the message over the DMCC $\left( \cX, p\left( y| x \right), \cY \right)$ corresponds to illegally uploading or
broadcasting the private content (possibly with some disturbance). Furthermore the statistical characterization of the
perturbation, which corresponds to the conditional p.m.f. $p\left( u| x \right)$, corresponds to the robust signal
hashing.
\end{remark}

The related error events will be shortly stated here to make the setup complete. Given an ACC code $\left( \cW, f^n, g^n, h^n \right)$ and the DMCC $\left( \cX, p\left( y| x \right), \cY \right)$,
the conditional probability of error, $\lambda_w$, conditioned on the transmitted message $W = w$ is given by
\begin{align*}
\lambda_w \eqdef \Pr \left( g \left( \bY^n \right) \neq w \, | \, W = w \right),
\end{align*}
where the probability is computed over the DMCC $\left( \cX, p\left( y| x \right), \cY \right)$. Maximal probability of error,
$\lambda^{\left( n \right)}$, is defined as, $\lambda^{(n)} = \max_{w \in \cW} \lambda_w$. Finally, the average probability of
error, $P_e^{\left( n \right)}$, is given by
\begin{align*}
P_e^{(n)} & \eqdef \Pr \left( \hat{W} \neq W \right) \\
 & = \sum_{w \in \cW} \Pr \left( g \left( \bY^n \right)
\neq w \, | \, W = w \right) \Pr \left( W = w \right) = 2^{-nR} \sum_{w \in \cW} \lambda_w.
\end{align*}

We point out that due to the nature of the problem setup, generation of the decoder codebook and generation of the channel output are
two independent events, given the encoder codebook. In fact, unlike ``broadcast-channel-like'' setups \cite{cover}, these two events does
not need to happen simultaneously. Hence for the random variables $X \in \cX$, $Y \in \cY$ and $U \in \cU$ with conditional
p.m.f.s $p\left(y | x\right)$ and $p \left(u | x \right)$, we have
$p\left(y , u | x\right) = p\left(y | x\right)p\left(u | x\right)$, which implies that $U \leftrightarrow X \leftrightarrow Y$
forms a Markov chain in the specified order.

The mismatch between deterministic codebook generation functions of encoder and decoder, namely $f \left( . \right)$ and
$h \left( . \right)$ in the respected order, constitutes the asymmetric nature of the problem. Hence, in accordance with
the problem definition, the functions $f \left( . \right)$ and $h \left( . \right)$ are only known by encoder and decoder,
respectively.

Finally, we define the necessary security conditions for the communication setup stated in Def.~\ref{def:ACC} (and shown in  Fig.~\ref{fig:setup})
to address the privacy issue in the content identification problem. For a given ACC
$\left( \cW, f^n, h^n, g^n \right)$,  we introduce the notion of ``security''
from an estimation theoretic perspective \cite{poor}. The asymmetric codebook generator
functions $f^n$ and $h^n$ are treated to be more ``secure'' if it is ``harder''
to estimate  the encoder codewords $\left\{ \bx^n \left( W \right) \right\}$
given the decoder codewords $\left\{ \bu^n \left( W \right) \right\}$. Note that
this approach is philosophically analogous to the desired ``approximate one-way'' property
of the robust hash functions. In order to achieve this task, we introduce the notion
of a ``codebook estimator function'' and a distortion metric that quantifies its performance.

For a given ACC $\left( \cW, f^n, h^n, g^n \right)$, a codebook estimator
function aims to estimate encoder codeword(s) given the corresponding decoder codeword(s) as:
\begin{itemize}
\item a {\em single-letter codebook estimator function} $\pi$ operates on individual
codeword elements such that  $\pi \, : \, \cU \rightarrow \cX$.
\item an {\em $n$-fold codebook estimator function} $\pi^n$ is a mapping such that
$\pi^n \, : \, \cU^n \rightarrow \cX^n$.
\end{itemize}

The performance of codebook estimation is quantified via the following:
\begin{itemize}
\item a {\em single-letter estimator distortion function} is a mapping
\[
d \, : \, \cX \times \cX \rightarrow \cR_d
\]
where $\cR_d \subseteq {\mathbb R}^+ \cup \left\{ 0 \right\}$;
the estimator distortion $d \left( \pi \left( u \right) , x \right)$ is
a measure of the cost of estimating $x$ by $\pi \left( u \right)$.
\item an {\em $n$-fold estimator distortion function} is a mapping
\[
d_n \, : \, \cX^n \times \cX^n \rightarrow \cR_d
\]
where $d_n \left( \hat{\bx}^n , \bx^n \right) \eqdef \frac{1}{n} \sum_{i=1}^n d \left( \hat{x}_i , x_i \right)$,
$\hat{\bx} \eqdef \pi^n \left( \bu^n \right)$.
\end{itemize}

An estimator distortion function is said to be ``bounded'' if the set of its values is bounded, i.e., there exists a
$D \in \mathbb R$ such that $\max_{\hat{x} \in \cX ,x \in \cX} d \left( \hat{x} , x \right) < D$. Given the codebook
estimator function $\pi^n$ and estimator distortion function $d_n$, we next define the $\alpha$-secure ACC codes which
satisfies certain security constraints.

\begin{definition}
An  ACC $\left( \cW, f^n, h^n, g^n \right)$ is said to be {\em $\alpha$-secure} if
\begin{align}
\min_{ \pi^n \, : \, \cU^n \rightarrow \cX^n} \mbox{E} \left[ d_n \left( \pi^n \left( \bU^n \right) , \bX^n \right) \right]  \geq \alpha ,
\label{eq:alpha-security}
\end{align}
where the expectation is with respect to the joint probability distribution of $\bX^n$ and $\bU^n$:
\[
\mbox{E} \left[ d_n \left( \pi^n \left( \bU^n \right) , \bX^n \right) \right]
=
\sum_{\bu^n \in \cU^n} \sum_{\bx^n \in \cX^n}
p \left( \bu^n , \bx^n \right) d_n \left( \pi^n \left( \bu^n \right) , \bx^n \right) .
\]
\label{def:alpha-secure}
\end{definition}

We emphasize that the minimization in \eqref{eq:alpha-security} is carried over the set of all possible $n$-fold estimator functions $\pi^n$ which has no restrictions, i.e., the estimator does not need to be sequential. The comprehensive approach in the Definition~\ref{def:alpha-secure} enables sustaining the security regardless of the type of the attack in the content identification problem \cite{poor}. Further note that the minimum in \eqref{eq:alpha-security} always exists since both the range and the domain
of $\pi_n$ are discrete, i.e., finite. Also, for all $\alpha_1 \leq \alpha_2$, an $\alpha_2$-secure ACC
is also $\alpha_1$-secure.

In this paper, we optimize the reliable communication rate of the system defined in Def.~\ref{def:ACC}, by sustaining
the security of the encoder codebook in an estimation theoretic perspective. Since the optimization is carried over the joint
p.m.f. $p \left(x , u \right)$, existence of the two generating functions $f \left( . \right)$ and $h \left( . \right)$
to maintain the predefined objectives of the problem is proved in this paper. Also as a side note, the statistical perturbation
between the encoder codebook $\cC_X$ and decoder codebook $\cC_U$, modeled
by the conditional p.m.f. $p \left(u | x \right)$, is assumed to be memoryless in this paper. Hence we study the memoryless
ACCs that are rigorously defined in the following.

\begin{definition}
Given a distribution $p \left( u , x \right)$ (with the domain $\cU \times \cX$),
a $\left( 2^{nR}, n \right)$  ACC $\left( \cW, f^n, h^n, g^n \right)$ is said to be {\em memoryless} if
$f^n$ and $h^n$  carry out {\em i.i.d. random codebook generation}
according to it. Specifically, both $f^n$ (resp. $h^n$) generate $2^{nR}$ codewords
$\left\{ \bx^n \left( w \right) \right\}_{w=1}^{2^{nR}}$ (resp. $\left\{ \bu^n \left( w \right) \right\}_{w=1}^{2^{nR}}$)
(each of which is of length-$n$) such that
\[
p \left(
\bu^n \left( 1 \right), \bu^n \left( 2 \right) , \ldots , \bu^n \left( 2^{nR} \right)  ,
\bx^n \left( 1 \right), \bx^n \left( 2 \right) , \ldots , \bx^n \left( 2^{nR} \right) \right)
=
\prod_{w=1}^{2^{nR}} \prod_{i=1}^n p \left( u_i \left( w \right) , x_i \left( w \right) \right).
\]
Such codes are termed as {\em memoryless asymmetric channel codes (MACC)}.
\label{def:MACC}
\end{definition}

We next derive the maximum reliable communication rate of the MACCs corresponding to the DMCC given in Fig.~\ref{fig:setup}
under certain security constraints to determine the fundamental limits of content identification.

\section{Memoryless Asymmetric Channel Codes Under Security Constraints - Capacity Results}
\label{sec:general-results}

In this section, we analyze and optimize the reliable communication rate of the communication system stated in
Section~\ref{ssec:problem-statement} under certain security constraints. We first introduce the maximum reliable
communication rate and the related definitions over memoryless asymmetric channel codes. We then give our fundamental
result in Theorem~\ref{thm:main-theorem} stating the channel coding results for the aforementioned setup, i.e.,
the maximum rate of error-free information transmission satisfying certain security constraints, which in
turn reveals the fundamental limits of the content identification problem. We then provide the achievability proof
(cf. Section~\ref{ssec:achievability}) and the converse proof (cf. Section~\ref{ssec:converse}) of the
Theorem~\ref{thm:main-theorem}.

\subsection{Capacity of MACC under Security Constraints}
\label{ssec:general-capacity}

In this section, we introduce the capacity, i.e., the maximum reliable communication rate,
of MACC under certain security constraints. Before giving the capacity definition
and the channel coding theorem, we need to define ``achievability'' notion, i.e., what
we mean when we say that a communication rate $R$ is achievable with a security
constraint $\alpha$. A pair $\left( R , \alpha \right)$ is said to be {\em achievable}
if there exists a sequence of $\left( 2^{nR} , n \right)$ MACC
$\left( \cW, f^n, h^n, g^n \right)$ such that
\begin{align}
\lim_{n \rightarrow \infty} P_e^{(n)} = \lim_{n \rightarrow \infty} \Pr \left[ W \neq g \left( \bY^n \right) \right] = 0 \,\,\,\,\,\,\,\,\mathrm{and}\,\,\,\,\,\,\,\,
\lim_{n \rightarrow \infty} \min_{\pi^n \, : \, \cU^n \rightarrow \cX^N}
\mbox{E} \left[ d \left( \pi^n \left( \bU^n \right) , \bX^n \right) \right] \geq \alpha.
\label{eq:achievable-security}
\end{align}

The {\em secure MACC region} is defined as the closure of all achievable $\left( R , \alpha \right)$ points and for any given
$\alpha$, the {\em secure MACC capacity} $C \left( \alpha \right)$ is the supremum of rates $R$ such that
$\left( R , \alpha \right)$ is in the secure MACC region. We next define a mathematical function of the MACC, which we call the 
{\em information secure MACC capacity}. Then, we state the main result of this paper by proving that the information 
secure MACC capacity is equal to the secure MACC capacity. Before we define the information secure MACC capacity, 
we first provide the definition of a feasible set of joint p.m.f.s $p \left( u , x \right)$, denoted by $P_\alpha$, which 
satisfies a certain security constraint. We then give the definition of the secure MACC capacity as the maximum
of the mutual information between the decoder's codeword and the communication channel output, where the maximization
is carried over $P_\alpha$.

\begin{definition}
For a given $\alpha$, the set of all joint distributions $p \left( u , x \right)$ that ensures the MACC $\left( \cW, f^n, h^n, g^n \right)$ to be $\alpha$-secure is given by
\begin{align}
P_\alpha \eqdef \left\{ p \left( u , x \right) \, \Big| \, \min_{\pi \, : \, \cU \rightarrow \cX}
\mbox{E} \left[ d \left( \pi \left( U \right) , X \right) \right] \geq \alpha \right\}.
\label{eq:p-alpha-def}
\end{align}
\label{def:p-alpha}
\end{definition}

We next define the information secure MACC capacity that employs a maximization over the set $P_\alpha$. We then introduce a theorem, which states that the proposed
information secure MACC capacity is equal to the secure MACC capacity, and constitutes the main result of this paper by providing the fundamental limit on the content identification problem.
\begin{definition}
For any given DMCC $\left( \cX, p\left( y| x \right), \cY \right)$, bounded distortion function
$d \left(\hat{x},x \right)$ and $\alpha$, if $P_\alpha$ (defined via (\ref{eq:p-alpha-def})) is non-empty, then
the {\em information secure MACC capacity} is defined as
\begin{align}
C^{(I)} \left( \alpha \right) \eqdef \max_{p \left( u , x \right) \in P_\alpha} I \left( U ; Y \right) .
\label{eq:information-capacity}
\end{align}
\label{def:information-capacity}
\end{definition}

We next show that the set $P_\alpha$ given in \eqref{eq:p-alpha-def} satisfies certain properties in the following proposition. We emphasize that by Proposition~\ref{prop:p-alpha-compactness}, the capacity function given in \eqref{eq:information-capacity} becomes definite, hence
this completes the definition of the information secure MACC capacity.

\begin{proposition}
For a given $\alpha$, the feasible set $P_\alpha$ satisfies the following properties,
\begin{itemize}
\item If $\alpha$ is such that $P_\alpha$ is the empty set, then $C \left( \alpha \right) = 0$.
\item If $\alpha$ is such that $P_\alpha$ is non-empty, then it is compact.
\end{itemize}
\label{prop:p-alpha-compactness}
\end{proposition}

\begin{proof}
See Appendix~\ref{app-1}.
\end{proof}

In the following remark, we state that the information secure MACC capacity is well-defined by employing the results of 
Proposition~\ref{prop:p-alpha-compactness}.

\begin{remark}
From proposition~\ref{prop:p-alpha-compactness}, we know that the set that $C^{(I)}$ is defined on $P_\alpha$ is compact.
Then, by using ``Maximum-Minimum Theorem'' \cite{analysis}, $I \left( U ; Y \right)$ is bounded and gets its minimum and maximum values
on the compact set $P_\alpha$, since $I \left( U ; Y \right)$ is a continuous function of $p \left( u , x \right)$.
Hence we can write maximum instead of supremum in the definition of $C^{(I)}$.
\label{rem:well-defined-capacity}
\end{remark}

We next introduce a theorem which states that the secure MACC capacity, i.e., the supremum of rates $R$ such that
$\left( R , \alpha \right)$ is in the secure MACC region, is equal to the information secure MACC capacity given in
\eqref{eq:information-capacity}. The following theorem constitutes the main contribution of this paper by providing
the maximum achievable rate of reliable communication for the communication system defined in Section~\ref{ssec:problem-statement}.
Therefore the theorem introduces the fundamental limits on the content identification problem, i.e., provides the upper bound
on the performance of any content identification algorithm satisfying certain level of privacy.

\begin{theorem}
{\em Secure Memoryless Asymmetric Channel Coding Theorem:}
For a discrete memoryless communications channel $\left( \cX, p\left( y| x \right), \cY \right)$ and for any given $\alpha$,
if $P_\alpha$ (defined via (\ref{eq:p-alpha-def})) is non-empty, then we have
\begin{align}
C \left( \alpha \right) = C^{(I)} \left( \alpha \right),
\label{eq:main-theorem}
\end{align}
where $C^{(I)} \left( \alpha \right)$ is given via (\ref{eq:information-capacity}). However,
if $P_\alpha$ is empty, then $C \left( \alpha \right) = 0$.
\label{thm:main-theorem}
\end{theorem}

Before presenting the proof of the theorem, we first introduce three lemmas that are used in the proof. The
first lemma provides the memoryless property of the joint p.m.f. $p \left(y , u\right)$.
The other two provide {\em essential} results about the ``separability'' of the $n$-fold codebook estimator function $\pi^n$.

\begin{lemma}
Given $p \left( \bx^n \right) = \prod_{i=1}^n p \left( x_i \right)$,
a discrete memoryless communications channel $\left( \cX , p \left( y | x \right) , \cY \right)$,
and  a $\left( 2^{nR} , n \right)$ MACC
$\left( \cX , \cC_X , p \left( u | x \right) , \cU, \cC_U \right)$,
we have
\begin{equation}
p \left( \by^n, \bu^n \right) = \prod_{i=1}^n p \left( y_i, u_i\right),
\label{eq:lemma-1}
\end{equation}
where $p \left( y, u \right) = \sum_{i=1}^n p(x) p(y|x) p(u|x)$.
\label{lem:lemma-1}
\end{lemma}

\begin{proof}
Following the definition and properties of joint p.m.f. \cite{prob}, we have
\begin{eqnarray}
p \left( \by^n, \bu^n \right) & = & \sum_{\bx^n \in \cX} p \left( \bx^n, \by^n, \bu^n \right) \nonumber \\
& = & \sum_{\bx^n \in \cX^n} p \left(\bx^n \right) p \left( \by^n | \bx^n \right) p \left( \bu^n | \bx^n \right) \label{eq:app2-1} \\
& = & \sum_{\bx^n \in \cX^n} \prod_{i=1}^n p \left(\bx_i \right) p \left( \by_i | \bx_i \right) p \left( \bu_i | \bx_i \right) \label{eq:app2-2} \\
& = & \prod_{i=1}^n \sum_{x_i \in \cX} p \left(x_i \right) p \left( y_i | x_i \right) p \left( u_i | x_i \right) \nonumber \\
& = & \prod_{i=1}^n \sum_{x_i \in \cX} p \left(x_i, y_i, u_i \right), \nonumber \\
& = & \prod_{i=1}^n p \left(y_i, u_i \right).
\end{eqnarray}
where \eqref{eq:app2-1} follows from that $U$, $X$ and $Y$ forms a Markov chain in the specified order, i.e., $U \leftrightarrow X \leftrightarrow Y$ and \eqref{eq:app2-2} follows since the communication channel and the perturbation are memoryless.
\end{proof}

\begin{lemma}
Given $n$ and the triplet $\left(X, U, V_1^n\right)$ such that $X \leftrightarrow U \leftrightarrow V_1^n$ where $\forall i$ $v_i \in \cU$, $u \in \cU$, $x \in \cX$, we have
\begin{align}
	\underset{\bpi \, : \, \cU^{n+1} \rightarrow \cX }{\min} \mbox{E} \left[d\left(\bpi\left(U,V_1^n\right),X\right)\right]
	=\underset{\pi \, : \, \cU \rightarrow \cX }{\min} \mbox{E} \left[d\left(\pi\left(U\right),X\right)\right]
	\label{eq:mapping-lemma}
\end{align}
\label{lem:mapping-lemma}
\end{lemma}
\begin{proof}
Let $\hat{X} \eqdef \bpi \left(U,V_1^n\right)$. Then,
\begin{align}
	\underset{\bpi \, : \, \cU^{n+1} \rightarrow \cX }{\min} \mbox{E}\left[d\left(\bpi\left(U,V_1^n\right),X\right)\right]
	= \underset{\bpi \, : \, \cU^{n+1} \rightarrow \cX }{\min} \mbox{E} \left[d\left(\hat{X},X\right)\right]
\end{align}
where the expectation is with respect to the joint p.m.f. of $\left(X, U, V_1^n\right)$.

From Bayesian estimation theory, this is equivalent to solving
\begin{align}
	\underset{\hat{x} \in \cX}{\min} \sum_{x \in \cX} d\left(\hat{x},x\right) p\left(x|u,v_1^n\right)
	\label{eq:bayesian-estimation}
\end{align}
for any given $U = u$, $V_1^n = v_1^n$. Note that the argument $f\left(x\right)$ is known as posterior risk conditioned on $\left(U,V_1^n\right) = \left(u,v_1^n\right)$. Since $X \leftrightarrow U \leftrightarrow V_1^n$, we have $p\left(x|u,v_1^n\right) = p\left(x|u\right)$, which implies \eqref{eq:bayesian-estimation} is equivalent to
\begin{align}
\underset{\hat{x} \in \cX}{\min} \sum_{x \in \cX} d\left(\hat{x},x\right) p\left(x|u\right).
\label{eq:bayesian-equivalent}
\end{align}
On the other hand, \eqref{eq:bayesian-equivalent} would be the problem resulting from solving the right hand side of \eqref{eq:mapping-lemma} via following analogous steps to the aforementioned procedure. Hence, this completes the proof.
\end{proof}

\begin{lemma}
Given $\left(U^n,X^n\right) \sim p\left(u^n,x^n\right) = \prod_{i=1}^{n} p\left(u_i|x_i\right)p\left(x_i\right)$, we have
\begin{align}
	\underset{\pi^n \, : \, \cU^n \rightarrow \cX^n }{\min} \mbox{E}\left[d^n\left(\pi^n\left(U^n\right),X^n\right)\right]
	= \frac{1}{n} \sum_{i=1}^{n} \underset{\pi \, : \, \cU \rightarrow \cX }{\min} \mbox{E}\left[d\left(\pi\left(U_i\right),X_i\right)\right].
\label{eq:seperable-estimation}
\end{align}
\label{lem:seperable-estimation}
\end{lemma}

Lemma~\ref{lem:seperable-estimation} states that the best $n$-fold codebook estimator function is separable, i.e., its performance
can be quantified via the performance of the best single-letter codebook estimator functions. Hence, it implies that in the content
identification problem, the performance of the best attack towards the private content can be determined via the performance of the
best attack towards a single entry of the content.
\begin{proof}
Let $\bpi^n = \underset{\pi^n \, : \,\cU^n \rightarrow \cX^n}{\arg\min} \mbox{E} \left[d^n\left(\pi^n\left(U^n\right),X^n\right)\right]$. Furthermore given $\bpi^n(.)$, define $\bpi_i \,:\, \cU^n \rightarrow \cX$ for all $i \in \{ 1,2, \ldots ,n\}$ such that $\bpi^n \left(u^n\right) = \left(\bpi_1\left(u^n\right),\bpi_2\left(u^n\right), \ldots ,\bpi_n\left(u^n\right)\right)$. Then, we have
\begin{eqnarray}
\underset{\pi^n \, :\, \cU^n \rightarrow \cX^n}{\min} \mbox{E} \left[d^n\left(\pi^n\left(U^n\right),X^n\right)\right]
& = &\mbox{E} \left[d^n\left(\bpi^n\left(U^n\right),X^n\right)\right] \nonumber \\
& = & \frac{1}{n} \sum_{i=1}^{n} \mbox{E} \left[d\left(\bpi_i\left(U^n\right),X_i\right)\right] \nonumber \\
& \geq & \frac{1}{n} \sum_{i=1}^{n} \underset{\tpi^n \, :\, \cU^n \rightarrow \cX}{\min}\mbox{E} \left[d\left(\tpi^n\left(U^n\right),X_i\right)\right] \label{eq:min-estimator}  \\
& = & \frac{1}{n} \sum_{i=1}^{n} \underset{\pi \, :\, \cU \rightarrow \cX}{\min}\mbox{E} \left[d\left(\pi\left(U_i\right),X_i\right)\right], \nonumber
\end{eqnarray}
where \eqref{eq:min-estimator} follows since per definition
\begin{align}
\underset{\tpi^n \, :\, \cU^n \rightarrow \cX}{\min}\mbox{E} \left[d\left(\tpi^n\left(U^n\right),X_i\right)\right]
\leq \mbox{E} \left[d\left(\bpi_i\left(U^n\right),X_i\right)\right]
\label{eq:min-perletter-estimation}
\end{align}
for all $\tpi_i$ and \eqref{eq:min-perletter-estimation} follows from Lemma~\ref{lem:mapping-lemma}.\\
Now let $\pi_i^{*} = \underset{\pi \, :\, \cU \rightarrow \cX}{arg\min} \, \mbox{E}\left[d\left(\pi\left(U_i\right),X_i\right)\right]$ for all $i \in \{ 1,2, \ldots ,n \}$ and subsequently define $\pi^{n,*} \, : \, \cU^n \rightarrow \cX^n$ such that $\pi^{n,*}\left(U^n\right) = \left[ \pi_1^{*},\pi_2^{*}, \ldots ,\pi_n^{*} \right]$. Then, we have
\begin{eqnarray}
\frac{1}{n} \sum_{i=1}^{n} \underset{\pi \, :\, \cU \rightarrow \cX}{\min}\mbox{E} \left[d\left(\pi\left(U_i\right),X_i\right)\right]
& = & \frac{1}{n} \sum_{i=1}^{n} \mbox{E} \left[d\left(\pi_i^{*}\left(U_i\right),X_i\right)\right] \nonumber  \\
& = & \mbox{E} \left[d^n\left(\pi^{n,*}\left(U^n\right),X^n\right)\right] \nonumber \\
& \geq & \underset{\pi^n \, :\, \cU^n \rightarrow \cX^n}{\min} \mbox{E} \left[d^n\left(\pi^n\left(U^n\right),X^n\right)\right]. \label{eq:seperability-ineq}
\end{eqnarray}
Combining \eqref{eq:min-perletter-estimation} and \eqref{eq:seperability-ineq} we get \eqref{eq:seperable-estimation}.
This completes  the proof.	
\end{proof}

Since we have the necessary lemmas, we next provide the proof of Theorem~\ref{thm:main-theorem}. The proof of Theorem~\ref{thm:main-theorem} is separated into two parts, i.e, achievability and converse proofs. In the achievability proof, we validate the achievability in the theorem, i.e., we show that every rate $R < C^{(I)} \left( \alpha \right)$ is achievable. In the converse proof, which concludes the proof of Theorem~\ref{thm:main-theorem}, we show that every achievable rate $R$ satisfies $R < C^{(I)} \left( \alpha \right)$. We next continue with the proof of the forward statement of Theorem~\ref{thm:main-theorem}.

\subsection{Achievability}
\label{ssec:achievability}
In this section, we provide a theorem which constitutes the achievability proof of Theorem~\ref{thm:main-theorem} by showing
that any rate below the information secure MACC capacity, i.e., $C^{(I)}$, is achievable. The following theorem emphasizes the
feasibility of the secure and reliable communication with a rate below the information secure MACC capacity. In the content
identification problem, it corresponds to the achievability of any performance below the fundamental limit for a certain level
of privacy.
\begin{theorem}
({\em Achievability})
For any given $\alpha$ such that $P_\alpha$ is non-empty, and for every rate $R < C^{(I)} \left( \alpha \right)$,
there exists a sequence of $\left( 2^{nR} , n \right)$ $\alpha$-secure MACCs with arbitrarily small maximal probability  of error
for sufficiently large $n$.
\label{thm:achievability}
\end{theorem}

\begin{proof}
We prove that for any $\alpha$ with $P_\alpha \neq \emptyset$, and any $R < C^{(I)} \left( \alpha \right)$, $\left(R, \alpha \right)$ pair is achievable by proving the existence of a sequence of $\alpha$-secure MACCs with rate $R$ satisfying the achievability conditions given in \eqref{eq:achievable-security}. In the {\em encoding} part, by choosing $p\left(u,x\right)$ from the set $P_\alpha$, encoder codebook $\cC_X$ and the decoder
codebook $\cC_U$ are generated as stated in Def.~\ref{def:MACC}. After choosing a message $w$ uniformly from the
message set $\cW$ (i.e., $\Pr\left(W = w \right) = 2^{-nR}$ for all $w \in \cW$), $f^n \left( w \right) = \bx^n(w)$  is
generated and transmitted over the DMCC $(\cX, p(y|x), \cY)$, resulting in $\bY^n$ such that
$\Pr \left( \bY^n = \by^n | \bx^n \left( w \right) \right) = \prod_{i=1}^n p \left( y_i | x_i \left( w \right) \right)$.
In the {\em decoding} part, we use jointly typical decoding. Note that $\left\{ u_i \left( W \right), Y_i \right\}_{i=1}^n$
pairs are independent of each other (cf. Lemma~\ref{lem:lemma-1}), where $\bY^n$ is the resulting communication channel
output corresponding to the message $W \in \cW$. If a unique $\hW \in \cW$ exists such that $\left(\bu^n \left( \hW \right),
\bY^n\right) \in A_\ep^{(n)} \left( U , Y \right)$, where $A_\ep^{(n)} \left( U , Y \right)$  is the $\ep$-jointly-typical
set \cite{cover}, defined as
\begin{align}
A_\ep^{(n)} \left( U , Y \right) \eqdef {\Bigg \{}  \left( \bu^n, \by^n \right) \; : \;
& |-\frac{1}{n}\log p\left(\bu^n\right) - H(U)| < \ep, \,
|-\frac{1}{n}\log p\left(\by^n\right) - H(Y)| < \ep,  \, \\
& |-\frac{1}{n}\log p\left(\bu^n,\by^n\right) - H(U,Y)| < \ep
{\Bigg \} }, \label{eq:typical-set}
\end{align}
where $p \left( u , y \right) = \sum_{x \in \cX} p \left( x \right) p \left( y | x \right)
p \left( u | x \right)$, then we declare $g \left( \bY^n \right) = \hW$. Otherwise, i.e., if such a $\hW \in \cW$ is not unique or does not exist, then we declare $g \left( \bY^n \right) = 0$. The error event is defined as
\begin{equation}
\cE \eqdef \left\{ \hW \neq W \right\}.
\label{eq:err-event}
\end{equation}

In order to prove that an $\left(R,\alpha\right)$ pair is achievable, we need to prove that the two achievability conditions
given in \eqref{eq:achievable-security} holds. Since $p\left(u,x\right)$ is chosen from the nonempty set $P_\alpha$, we have
\begin{align}
\min_{\pi \, : \, \cU \rightarrow \cX} \mbox{E} \left[ d \left( \pi \left( U_i \right) , X_i \right) \right] \geq \alpha
\label{eq:eq-0}
\end{align}
for all $i \in \{ 1,2, \ldots ,n\}$. Hence \eqref{eq:eq-0} yields
\begin{eqnarray}
\underset{\pi^n \, :\, \cU^n \rightarrow \cX^n}{\min} \mbox{E} \left[d^n\left(\pi^n\left(U^n\right),X^n\right)\right]
& = &  \frac{1}{n} \sum_{i=1}^{n} \underset{\pi \, : \, \cU \rightarrow \cX }{\min} \mbox{E}\left[d\left(\pi\left(U_i\right),X_i\right)\right],\label{eq:eq-1} \\
& \geq &  \frac{1}{n} \sum_{i=1}^{n} \alpha,\nonumber \\
& = & \alpha, \label{eq:eq-2}
\end{eqnarray}
where \eqref{eq:eq-1} follows from Lemma~\ref{lem:seperable-estimation}. Since \eqref{eq:eq-2} is true for all $n>0$, we have
\begin{align}
\lim_{n \rightarrow \infty} \min_{\pi^n \, : \, \cU^n \rightarrow \cX^N}
\mbox{E} \left[ d \left( \pi^n \left( \bU^n \right) , \bX^n \right) \right] \geq \alpha.
\label{eq:eq-3}
\end{align}
We now prove that the other condition of the achievability is satisfied, too, i.e., we show that
\begin{align}
\lim_{n \rightarrow \infty} P_e^{(n)} = \lim_{n \rightarrow \infty} \Pr \left[ W \neq g \left( \bY^n \right) \right] = 0. \nonumber
\end{align}
From the definition of average probability of error of an MACC we have
\begin{align}
P_e^{(n)} \eqdef \Pr \left( \hat{W} \neq W \right)
= \sum_{w \in \cW} \Pr \left( g \left( \bY^n \right) \neq w \, | \, W = w \right) \Pr \left( W = w \right)
= 2^{-nR} \sum_{w \in \cW} \lambda_w \nonumber.
\end{align}
We calculate the average probability of error by taking the average over all decoder codebooks,
\begin{eqnarray}
\Pr\left( \cE \right) & = & \sum_{\cC_U} \Pr\left( \cC_U \right) P_e^{(n)}\left( \cC_U \right) \nonumber \\
 & = & \frac{1}{2^{nR}}\sum_{w=1}^{2^{nR}}\sum_{\cC_U}\Pr\left(\cC_U \right)\lambda_w\left( \cC_U\right) \nonumber \\
 & = & \sum_{\cC_U} \Pr\left( \cC_U \right)\lambda_1\left( \cC_U\right), \label{eq:average-prob-err-2} \\
 & = & \Pr \left( \cE | W=1 \right), \nonumber
\end{eqnarray}
where \eqref{eq:average-prob-err-2} follows since the codebook construction is symmetric and does not depend on the particular message that was sent. Hence, after this point, w.l.o.g., we assume the message $W=1$ was sent.

Let $E_i$ denotes the event that the codeword $\bX^n(i)$ and $\bY^n$ are jointly typical, where $\bY^n$ is the resulting output of the communication system corresponding to the codeword $\bX^n(1)$. Hence we can define $E_i$ as,
\begin{align}
E_i \eqdef \left\{ \left( \bu^n \left( i \right), \bY^n \right) \in A_\epsilon^{(n)}\right\}, \,\,\,\,\, i \in \left\{1, 2, \dots, 2^{nR} \right\}.
\nonumber
\end{align}
Then, we can write the average probability of error as
\begin{eqnarray}
P_e^{(n)} & = & \Pr \left( \cE | W=1 \right) \nonumber \\
& = & \Pr \left( E_1^c \cup E_2 \cup E_3 \dots E_{2^{nR}} | W=1 \right) \label{eq:average-prob-err-4} \\
& \leq & \Pr \left( E_1^c | W=1 \right) + \sum_{i=2}^{2^{nR}} \Pr \left( E_i | W=1 \right) \label{eq:average-prob-err-5} \\
& \leq & \epsilon + \sum_{i=2}^{2^{nR}} 2^{-n\left( I\left( U;Y \right) - 3\ep \right)}, \label{eq:average-prob-err-6}
\end{eqnarray}
where \eqref{eq:average-prob-err-4} follows since an error occurs only if either the transmitted codeword is not jointly typical with the received sequence, i.e., $E_1^c$ occurs, or the received sequence is jointly typical with a wrong codeword, i.e., $E_i$ occurs for $i \in \{2, 3, \dots, 2^{nR}\}$. Also \eqref{eq:average-prob-err-5} follows from the union bound and \eqref{eq:average-prob-err-6} follows from the joint AEP theorem \cite{cover} since $\bu^n(i)$ and $\bu^n(1)$ are independent for $i \neq 1$.

Finally we can write
\begin{eqnarray*}
P_e^{(n)} & \leq & \epsilon + \sum_{i=2}^{2^{nR}} 2^{-n\left( I\left( U;Y \right) - 3\ep \right)} \\
& = & \epsilon + (2^{nR}-1)2^{-n\left( I\left( U;Y \right) - 3\ep \right)} \\
& \leq & \epsilon + 2^{-n\left( I\left( U;Y \right) - 3\ep - R \right)} \\
& \leq & 2\epsilon,\label{eq:average-prob-err-7}
\end{eqnarray*}
for sufficiently large $n$ and $R < I\left( U;Y \right) - 3\ep$. Since for every rate $R < I\left( U;Y \right)$,
we can find $\ep > 0$ and a sufficiently large $n$ such that \eqref{eq:average-prob-err-7} holds. Combining
\eqref{eq:average-prob-err-7} and \eqref{eq:eq-3}, we have that for any given $\alpha$ such that $P_\alpha$ is
non-empty and for every rate $R < C^{(I)} \left( \alpha \right)$, the pair $(R,\alpha)$ is achievable. Note that as
mentioned in Remark~\ref{rem:well-defined-capacity}, we can choose $p\left(u,x\right)$ from the compact set $P_\alpha$
so as to maximize $I\left( U;Y \right)$ and then the condition $R < I\left( U;Y \right)$ can be replaced by
the achievability condition $R < C^{(I)} \left( \alpha \right)$.
\end{proof}
We next provide the proof of the converse statement of Theorem~\ref{thm:main-theorem}.

\subsection{Converse}
\label{ssec:converse}
In this section, we introduce a theorem providing the converse proof of the Theorem~\ref{thm:main-theorem} by showing 
that if a communication rate is achievable, then this rate should be below the information secure MACC capacity, i.e., $C^{(I)}$. The following
theorem emphasizes that the information secure MACC capacity constitutes an upper bound for the rate of  secure and reliable
communication. In the content identification problem, this provides the fundamental limit of successful anti-piracy search with a
certain level of privacy.
\begin{theorem}
({\em Converse})
For any given $\alpha$ such that $P_\alpha$ is non-empty and for any $\left( 2^{nR}, n \right)$
$\alpha$-secure MACC with $\lambda^{(n)} \rightarrow 0$, we have $R<C^{(I)} \left( \alpha \right)$.
\label{thm:converse}
\end{theorem}
\begin{proof}
We begin the proof with a lemma that states the necessity of one-to-one property of the deterministic decoder
codebook generator function $h^n$ for error-free MACCs.
\begin{lemma}
For any given $\left( 2^{nR} , n \right)$ MACC $\left( \cW, f^n, h^n, g^n \right)$ with $\lambda^{(n)} \rightarrow 0$,
$h^n$ is necessarily a one-to-one mapping.
\label{lem:lemma-2}
\end{lemma}
\begin{proof}
Following the proof in Appendix III of \cite{Alt08}, the proof follows.
\end{proof}
Lemma~\ref{lem:lemma-2} holds for general MACCs, hence it also applies for our setup, i.e.,
i.i.d. MACCs.

Since the transmitted message, i.e., $W$, and the communication channel
output, i.e., $\bY^n$, have a joint distribution and the decoder output, i.e., $\hW$, is a function of $\bY^n$,
then $W$, $\bY^n$ and $\hW$ form a Markov chain in the specified order, i.e.,
$W \leftrightarrow \bY^n \leftrightarrow \hW$.
Similarly, since $h^n \left( W \right) = \bu^n \left( W \right)$ is a function of $W$, then
$\bu^n(W)$, $W$ and $\bY^n$ form a Markov chain
in the that order, i.e., $\bU^n \leftrightarrow W \leftrightarrow \bY^n$, where $\bU^n$ denotes $\bu^n \left( W \right)$
notational simplicity. Combining these two Markov chains
yields that $\bU^n$, $W$, $\bY^n$ and $\hW$ form a Markov chain in the specified order, i.e.,
$\bU^n \leftrightarrow W \leftrightarrow \bY^n \leftrightarrow \hW$. Lemma~\ref{lem:lemma-2}
yields that $h^n \left( \cdot \right)$ is one-to-one, which further yields from the previous Markov chain that
$W$, $\bU^n$, $\bY^n$ and $\hW$ form a Markov chain in that order, i.e.,
$W \leftrightarrow \bU^n \leftrightarrow  \bY^n \leftrightarrow \hW$.

We continue to examine the pair $\bU^n$, $\bY^n$ such that
\begin{eqnarray}
p \left( \by^n | \bu^n \right) & = & \frac{p \left( \by^n, \bu^n \right)}{p \left( \bu^n \right)} \nonumber \\
 & = & \frac{\prod_{i=1}^n p \left( y_i, u_i \right)}{p \left( \bu^n \right)} \label{eq:discrete-case-converse-0.25} \\
 & = & \frac{\prod_{i=1}^n p \left( y_i, u_i \right)}{\prod_{i=1}^n \left[ \sum_{x_i} p(u_i|x_i) p(x_i) \right]} \label{eq:discrete-case-converse-0.5} \\
 & = & \prod_{i=1}^n p(y_i|u_i), \label{eq:discrete-case-converse-0.75}
\end{eqnarray}
where \eqref{eq:discrete-case-converse-0.25} follows from Lemma~\ref{lem:lemma-1}, \eqref{eq:discrete-case-converse-0.5} follows
from the definition of MACC and that
$p\left(\bx^n\right) = \prod_{i=1}^n p(x_i)$ and
\eqref{eq:discrete-case-converse-0.75} follows from
combining  $p(y,u) = \sum_{x_i} p(y|x)p(u|x)p(x)$
and  $p(y_i | u_i) = \frac{p \left( y_i, u_i \right)}{\sum_{x_i} p(u_i|x_i) p(x_i)}$
due to Bayes' rule.

Since the given MACC is $\alpha$-secure, we have
\begin{eqnarray}
	\alpha & \leq & \min_{ \pi_n \, : \, \cU^n \rightarrow \cX^n} \mbox{E} \left[ d_n \left( \pi^n \left( \bU^n \right) , \bX^n \right) \right] \nonumber \\
	& = & \frac{1}{n} \sum_{i=1}^{n} \underset{\pi \, : \, \cU \rightarrow \cX }{\min} \mbox{E}\left[d\left(\pi\left(U_i\right),X_i\right)\right] \label{eq:security-eq-1} \\
	& = & \mbox{E}\left[d\left(\pi\left(U_i\right),X_i\right)\right], \forall \, i \in \{1, 2, \ldots ,n\} ,\label{eq:security-eq-2}
\end{eqnarray}
where \eqref{eq:security-eq-1} follows from Lemma~\ref{lem:seperable-estimation}.

While we have \eqref{eq:discrete-case-converse-0.75}, \eqref{eq:security-eq-2} and the Markov chain,
$W \leftrightarrow \bU^n \leftrightarrow  \bY^n \leftrightarrow \hW$, we now continue with the
following chain of inequalities
\begin{eqnarray}
nR & = & H\left(W\right) \label{eq:discrete-case-converse-1} \\
 & = & I\left(\hW; W\right) + H\left(W | \hW\right), \nonumber \\
 & \leq & I\left(\bU^n; \bY^n\right) + \left( 1 + nR P_e^{\left(n\right)}\right) \label{eq:discrete-case-converse-2} \\
 & = & H\left( \bY^n \right) - \sum_{i=1}^n H\left( Y_i | U_i\right) + \left( 1 + nR P_e^{(n)}\right) \label{eq:discrete-case-converse-3} \\
 & = & \left( 1 + nR P_e^{(n)}\right) + \sum_{i=1}^n \left( H(Y_i) - H(Y_i | U_i) \right) \label{eq:discrete-case-converse-4} \\
 & = & \left( 1 + nR P_e^{(n)}\right) + \sum_{i=1}^n I(U_i ; Y_i) \label{eq:discrete-case-converse-5} \\
 & \leq & \left( 1 + nR P_e^{(n)}\right) + \sum_{i=1}^n C^{(I)}[\min_{\pi \, : \, \cU \rightarrow \cX} E(d_{e}(\pi(U_i) , X_i))] \label{eq:discrete-case-converse-6} \\
& \leq & \left( 1 + nR P_e^{(n)}\right) + \sum_{i=1}^n C^{(I)}(\alpha) \label{eq:discrete-case-converse-7} \\
& = & \left( 1 + nR P_e^{(n)}\right) + nC^{(I)}(\alpha), \label{eq:discrete-case-converse-8}
\end{eqnarray}
where \eqref{eq:discrete-case-converse-1} follows since $W$ is
uniformly distributed over $\cW$, \\
\eqref{eq:discrete-case-converse-2} follows using Fano's inequality (the second term)
and the data processing inequality  (the first term) by recalling that
$W \leftrightarrow \bU^n \leftrightarrow \bY^n \leftrightarrow \hW$
forms a Markov chain in the specified order, \\
\eqref{eq:discrete-case-converse-3} follows using \eqref{eq:discrete-case-converse-0.75},\\
\eqref{eq:discrete-case-converse-4} follows since $p \left( \by^n
\right) = \prod_{i=1}^n \sum_{x_i} p(y_i | x_i) p(x_i) = \prod_{i=1}^n p \left( y_i \right)$, the communication channel is memoryless and $p \left(
\bx^n \right) = \prod_{i=1}^n p (x_i)$, which implies that
$H\left(\bY^n\right)  = \sum_{i=1}^n H(Y_i)$,\\
\eqref{eq:discrete-case-converse-5} follows using the definition of
mutual information, \\
\eqref{eq:discrete-case-converse-6} follows from the definition of information secure MACC capacity,\\
\eqref{eq:discrete-case-converse-7} follows from \eqref{eq:security-eq-2}
and  $C^{(I)}\left(\alpha \right)$ is a nonincreasing function of $\alpha$,

Using \eqref{eq:discrete-case-converse-8} and noting that $\lambda^{(n)} \rightarrow 0$ implies $P_e^{(n)} \rightarrow 0$, we have
\begin{eqnarray}
	R & \leq & \frac{1}{n} + RP_e^{(n)} + C^{(I)}\left(\alpha\right) \nonumber \\
	& \leq & \ep + C^{(I)}\left(\alpha\right), \label{eq:discrete-case-converse-9}
\end{eqnarray}
for any $\ep > 0$ and sufficiently large $n$, where \eqref{eq:discrete-case-converse-9} follows since $P_e^{(n)} \rightarrow 0$ and $1/n \leq \ep$ for sufficiently large $n$. Therefore \eqref{eq:discrete-case-converse-9} implies $R < C^{(I)}\left(\alpha\right)$, which concludes the proof.
\end{proof}

\section{Binary Alphabet Case With Error Probability Based Security Constraints}
\label{sec:binary}

In this section, we consider a special case of interest where the codewords of the encoder are drawn from a binary alphabet, i.e., $\cX=\cY=\cU=\{0,1\}$, the communication channel is a binary symmetric channel with crossover probability $p_1$ and the perturbation distribution is binary symmetric distribution with parameter $p_2$. We introduce a closed form expression of the information secure MACC capacity, i.e.,
\begin{align}
C^{(I)} \left( \alpha \right) = \max_{p \left( u , x \right) \in P_\alpha} I \left( U ; Y \right) ,
\label{eq:info-capacity}
\end{align}
for this binary alphabet case. For this special case of interest, we assume that the distortion function is a Hamming distortion
\begin{align*}
d(x,\hat{x}) = \genfrac{\{}{\}}{0pt}{}{0 \ \ \mathrm{if} \ \ \ x=\hat{x}}{1 \ \ \mathrm{if} \ \ \ x\neq\hat{x}}
\end{align*}
which is a well-known distortion measure extensively used in the literature \cite{cover}. Note that since the Hamming distortion satisfies $\mbox{E} \left[ d \left( \left( \hat{X} \right) , X \right) \right] = \mathrm{Pr} \left ( \hat{X} \neq X \right)$, then the security constraint $\mbox{E} \left[ d \left( \left( \hat{X} \right) , X \right) \right] \geq \alpha$ becomes $\mathrm{Pr} \left ( \hat{X} \neq X \right)$. Hence the the definition of $P_{\alpha}$ becomes
\begin{align}
P_\alpha \eqdef \left\{ p \left( u , x \right) \, \Big| \, \min_{\pi \, : \, \cU \rightarrow \cX}
\mathrm{Pr} \left ( \pi \left( U \right) \neq X \right) \geq \alpha \right\}.
\label{eq:bin_palpha}
\end{align}
Furthermore, given $\left(U , X \right) \sim p\left( u, x \right)$, the estimator $\hat{X}_{\mathrm{MAP}} \eqdef \underset{x}{\arg\max}\, p \left( x | u \right) = \underset{x}{\arg\max}\, p \left(u | x \right)p \left( x \right)$ minimizes $\mathrm{Pr} \left ( \hat{X} \left( U \right) \neq X \right)$ in \eqref{eq:bin_palpha} and known as the MAP estimator. Then the security constraint $p(u,x) \in P_\alpha$ in \eqref{eq:info-capacity} becomes
\begin{align}
\mathrm{Pr} \left ( \hat{X}_{\mathrm{MAP}} \neq X \right) \geq \alpha.
\label{eq:sec-cons-1}
\end{align}
Note that in order to make the problem valid and meaningful, we necessarily need to have $\min\{p_2,1-p_2\} \geq \alpha$, otherwise one of the trivial estimators $\hat{X} = U$ and $\hat{X} = U \oplus 1$ violates the security constraint.

For notational clarity, let us denote
\begin{align*}
p \left( u=0 , x=0 \right) = \gamma_0,\,\,\,\,\,\, p \left( u=0 , x=1 \right) = \gamma_1,\\
p \left( u=1 , x=0 \right) = \beta_0,\,\,\,\,\,\, p \left( u=1 , x=1 \right) = \beta_1,
\end{align*}
where $\gamma_0 + \gamma_1 + \beta_0 + \beta_1 = 1$. Then clearly, if $U = 0$, then
\begin{align*}
\hat{X}_{\mathrm{MAP}} = \genfrac{\{}{\}}{0pt}{}{1 \ \ \mathrm{if} \ \ \gamma_0 < \gamma_1}{0 \ \ \mathrm{if} \ \ \gamma_0 > \gamma_1}
\end{align*}
and if $U = 1$, then
\begin{align*}
\hat{X}_{\mathrm{MAP}} = \genfrac{\{}{\}}{0pt}{}{1 \ \ \mathrm{if} \ \ \beta_0 < \beta_1}{0 \ \ \mathrm{if} \ \ \beta_0 > \beta_1}.
\end{align*}
Hence we have
\begin{eqnarray}
\mathrm{Pr} \left ( \hat{X}_{\mathrm{MAP}} \neq X \right) &=& \mathrm{Pr} \left( \hat{X}_{\mathrm{MAP}} \neq X, U=0 \right) + \mathrm{Pr} \left( \hat{X}_{\mathrm{MAP}} \neq X, U=1 \right) \nonumber \\
&=& \min\{ \gamma_0,\gamma_1\} + \min\{ \beta_0,\beta_1\}. \label{eq:sec-cons-2}
\end{eqnarray}
Combining \eqref{eq:sec-cons-2} and the security constraint \eqref{eq:sec-cons-1}, information secure MACC capacity in \eqref{eq:info-capacity} becomes
\begin{align}
C^{(I)} \left( \alpha \right) = \max_{\min\{ \gamma_0,\gamma_1\} + \min\{ \beta_0,\beta_1\} \geq \alpha} I \left( U ; Y \right) .	
\label{eq:info-capacity-2}
\end{align}

We continue with a lemma which is given and proved in \cite{Alt08}, which
states that
\begin{align}
I(U;Y) =  H(X \oplus Z_1 \oplus Z_2) - H(p_1 + p_2 -2 p_1 p_2),
\label{eq:alt_info}
\end{align}
where $Z_1$ and $Z_2$ are two binary random variables with
$Pr \left ( Z_1 = 1 \right) = p_1$, $Pr \left ( Z_1 = 0 \right) = 1 - p_1$ and $Pr \left ( Z_2 = 1 \right) = p_2$,
$Pr \left ( Z_2 = 0 \right) = 1 - p_2$ and $H \left( p \right)$ is the {\em binary entropy} function (with an abuse of notation), i.e., $H \left( p \right)
\eqdef - p \log p - \left( 1 - p \right) \log \left( 1 - p \right)$ for $p \in \left[ 0 , 1 \right]$. Combining \eqref{eq:alt_info} and \eqref{eq:info-capacity-2}, we next find a closed form expression of the information secure MACC capacity. Note that the information secure MACC capacity is the maximum mutual information $I(U;Y)$, where the maximization carried over $p(u,x)$ satisfying $\min\{ \gamma_0,\gamma_1\} + \min\{ \beta_0,\beta_1\} \geq \alpha$. Here, instead of analytically maximizing the mutual information, we find an upper bound on and prove that the upper bound is achievable. Since the binary entropy is upper bounded by $1$, we have
\begin{eqnarray}
I(U;Y) &=&  H(X \oplus Z_1 \oplus Z_2) - H(p_1 + p_2 -2 p_1 p_2),\\
&\leq& 1 - H(p_1 + p_2 -2 p_1 p_2).\label{eq:max-info}
\end{eqnarray}
To achieve the equality in \eqref{eq:max-info}, we choose $X$ as bernoulli $\frac{1}{2}$, hence we guarantee $X \oplus Z_1 \oplus Z_2$
to be bernoulli $\frac{1}{2}$, i.e., $H(X \oplus Z_1 \oplus Z_2) = 1$ and achieve the equality. By choosing the joint p.m.f. $p(u,x)$ as
\begin{align}
\gamma_0 = \beta_1 = \frac{1- p_2}{2} \nonumber\\
\gamma_1 = \beta_0 = \frac{p_2}{2},
\label{eq:ach-pmf}
\end{align}
we obtain
\begin{align}
p \left( x=0 \right) = p \left( u=0 , x=0 \right) + p \left( u=1 , x=0 \right) = \gamma_0 + \beta_0 = \frac{1- p_2}{2} + \frac{p_2}{2} = \frac{1}{2} \nonumber\\
p \left( x=1 \right) = p \left( u=0 , x=1 \right) + p \left( u=1 , x=1 \right) = \gamma_1 + \beta_1 = \frac{p_2}{2} + \frac{1 - p_2}{2} = \frac{1}{2}. \label{eq:bernoulli}
\end{align}
Hence we have $X$ to be bernoulli $\frac{1}{2}$. Furthermore \eqref{eq:ach-pmf} yields
\begin{align}
& p \left( u=1 | x=0 \right) = p \left( u=0 | x=1 \right) = 2\beta_0 = 2\gamma_1 = \frac{2p_2}{2} = p_2 \nonumber\\
& p \left( u=0 | x=0 \right) = p \left( u=1 | x=1 \right) = 2*\beta_1 = 2\gamma_0 = \frac{2(1-p_2)}{2} = 1-p_2, \label{eq:pert_dist}
\end{align}
which implies that the perturbation distribution is symmetric with $p_2$. Moreover the security constraint holds since
\begin{align*}
\mathrm{Pr} \left ( \hat{X} \neq X \right) = \min\{ \gamma_0,\gamma_1\} + \min\{ \beta_0,\beta_1\} = \min\{p_2,1-p_2\}
\end{align*}
and we have $\min\{p_2,1-p_2\} \geq \alpha$. Combining \eqref{eq:bernoulli}, \eqref{eq:pert_dist} and $\min\{p_2,1-p_2\} \geq \alpha$,
we conclude that by choosing $p(u,x)$ as in \eqref{eq:ach-pmf}, the information secure MACC capacity in \eqref{eq:info-capacity-2} becomes
\begin{align}
C^{(I)} \left( \alpha \right) = \max_{\min\{p_2,1-p_2\} \geq \alpha} 1 - H(p_1 + p_2 -2 p_1 p_2).	
\label{eq:info-capacity-3}
\end{align}

\begin{figure}
\centering \epsfxsize 4in
 \epsfbox{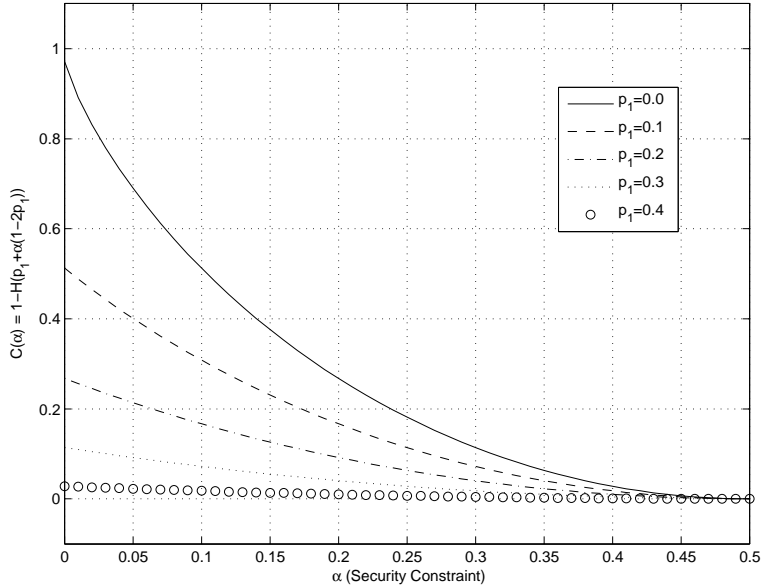}
 \caption{\small The information secure MACC capacity as a function of security constraint $\alpha$ for the binary alphabet case.}
\label{fig:results}
\end{figure}
Since $H(p_1 + p_2 -2 p_1 p_2)$ in \eqref{eq:info-capacity-3} is symmetric around $1/2$, we can assume that $0 \leq p_1,p_2 \leq 1/2$. Note that the binary entropy function $H(p)$ is monotonically increasing for $0 \geq p \leq 1/2$ and also $p_1 + p_2 -2 p_1 p_2$ is monotonically increasing for $0 \geq p_2 \leq 1/2$. Therefore the maximum of the function $1 - H(p_1 + p_2 -2 p_1 p_2)$, under the constraint that $\min\{p_2,1-p_2\} \geq \alpha$, is achieved when $p_2 = \alpha$, hence \eqref{eq:info-capacity-3} becomes
\begin{align}
C^{(I)} \left( \alpha \right) = 1 - H(p_1 + \alpha -2 p_1 \alpha) = 1 - H(p_1 + \alpha(1 - 2 p_1)),
\label{eq:bin-cap}
\end{align}
yielding a closed form expression for the information secure MACC capacity.

The information secure MACC capacity as a function of the security constraint $\alpha$ for various values of
$p_1$ is shown in Fig.~\ref{fig:results}. In consistent with the practical setup, we observe a  trade-off between
the security and the performance of the communication setup. To maximize the security, or $\alpha$ in our case,
one needs to choose $\alpha = \frac{1}{2}$. Then the information secure MACC capacity becomes
$C^{(I)}\left( \frac{1}{2} \right) = 1 - H(p_1 + 1/2 - p_1) = 0$. On the other hand, to maximize the capacity
of the system, regardless of the channel, we have to choose $\alpha = 0$, i.e., ``zero security''.

\section{Conclusions}
\label{sec:conclusions}

In this paper, we studied the content identification problem: (1)
where a rights-holder company desires to keep track of illegal uses of
its commercial content, (2) by utilizing resources of a security
company, (3) while securing the privacy of its content, from an
information theoretic perspective. The content identification is
modelled as a communication problem using a asymmetric codebooks, where
the commercial content of the rights-holder company corresponds to the
codebook of an encoder and the hash values of the content (made
available to the security company) corresponds the codebook of a
decoder. The privacy issue in the content identification is modelled
by adding certain security constraints to this communication setup to
prevent estimation of the encoder codewords given the decoder
codewords. By this modeling, the proposed problem of reliable
communication with asymmetric codebooks with security constraints
provided the fundamental limits of the content identification
problem. Under this framework, we introduced an information capacity
and proved that this capacity is equal to the operation capacity of
the system under i.i.d. encoder codewords, yielding the fundamental
limits for content identification. As a well known and widely studied
framework, we evaluated the capacity for a binary symmetric channel
and provided closed form expressions.

\appendices
\section{Proof of Proposition~\ref{prop:p-alpha-compactness}}
\label{app-1}
\setcounter{equation}{0}
\renewcommand{\theequation}{I-\arabic{equation}}

We begin with the proof of the first property that if $\alpha$ is such that $P_\alpha$ is the empty set, then $C \left( \alpha \right) = 0$, by way of contradiction. Assume that there exists an $\alpha_0$ such that the set $P_{\alpha_0}$ is empty, while $C \left( \alpha_0 \right) > 0$. Hence there exists a rate $R > 0$ such that $(R,\alpha_0)$ pair is achievable. Then, by using the definition of achievability, we have a sequence of
$\left( 2^{nR} , n \right)$ MACC $\left( \cW, f^n, h^n, g^n \right)$ such that
\begin{align}
\lim_{n \rightarrow \infty} \min_{\pi^n \, : \, \cU^n \rightarrow \cX^N}
\mbox{E} \left[ d \left( \pi^n \left( \bU^n \right) , \bX^n \right) \right] \geq \alpha_0.
\label{eq:app-1}
\end{align}
By using the Lemma~\ref{lem:seperable-estimation}, \eqref{eq:app-1} yields
\begin{align}
\lim_{n \rightarrow \infty} \frac{1}{n} \sum_{i=1}^{n} \underset{\pi \, : \, \cU \rightarrow \cX }{\min} \mbox{E}\left[d\left(\pi\left(U_i\right),X_i\right)\right] \geq \alpha_0 .
\label{eq:app-2}
\end{align}
Since the codewords of the encoder and the decoder are realizations of an i.i.d. process, then \eqref{eq:app-2} becomes
\begin{align}
\underset{\pi \, : \, \cU \rightarrow \cX }{\min} \mbox{E}\left[d\left(\pi\left(U\right),X\right)\right] \geq \alpha_0,
\nonumber
\end{align}
which means that $P_{\alpha_0}$ is not empty by definition of the set $P_\alpha$,
which concludes the proof of the first property that if $\alpha$ is such that $P_\alpha$
is the empty set, then $C \left( \alpha \right) = 0$.

To prove that the nonempty set $P_{\alpha}$ is compact, we use Heine-Borel theorem \cite{analysis} which states that any subset $S$ of $\bbR^n$ is compact if and only if it is bounded and closed.
Since $P_{\alpha}$ is the set of joint distributions of the realizations of $X$ and $U$, then its elements obey the rules of probability, hence the set is bounded by $[0,1]^{|\cX \times \cX|}$ for the discrete finite case.
To prove that $P_{\alpha}$ is closed, we prove that for every convergent sequence $p_k \in P_{\alpha}$, the limit lies in $P_{\alpha}$ \cite{analysis}. 

Hence, assume that there exists a convergent sequence $p_k \rightarrow p$ in $P_{\alpha}$. By using the definition of limit, $\forall \epsilon >0$ there exists an $N$ such that $\forall n \geq N$ we have $|p_n - p| < \epsilon$. Note that $p_k(u,x)$ is in $P_{\alpha}$ if and only if
\begin{align}
\lim_{n \rightarrow \infty} \min_{\pi^n \, : \, \cU^n \rightarrow \cX^N}
\mbox{E}_{p_k} \left[ d \left( \pi^n \left( \bU^n \right) , \bX^n \right) \right] \geq \alpha.
\label{eq:app-3}
\end{align}
Furthermore, for a given n-fold codebook estimator function $\pi^n$, we have
\begin{eqnarray}
\mbox{E}_p \left[ d_n \left( \pi^n \left( \bU^n \right) , \bX^n \right) \right]
& = &
\sum_{\bu^n \in \cU^n} \sum_{\bx^n \in \cX^n}
p \left( \bu^n , \bx^n \right) d_n \left( \pi^n \left( \bu^n \right) , \bx^n \right) \nonumber \\
& > & \sum_{\bu^n \in \cU^n} \sum_{\bx^n \in \cX^n}
(p_k \left( \bu^n , \bx^n \right) - \epsilon) d_n \left( \pi^n \left( \bu^n \right) , \bx^n \right) \nonumber \\
& = & \sum_{\bu^n \in \cU^n} \sum_{\bx^n \in \cX^n}
(p_k \left( \bu^n , \bx^n \right) d_n \left( \pi^n \left( \bu^n \right) , \bx^n \right) \nonumber \\
& \, & - \sum_{\bu^n \in \cU^n} \sum_{\bx^n \in \cX^n}
\epsilon d_n \left( \pi^n \left( \bu^n \right) , \bx^n \right) \nonumber \\
& > & \sum_{\bu^n \in \cU^n} \sum_{\bx^n \in \cX^n}
(p_k \left( \bu^n , \bx^n \right) d_n \left( \pi^n \left( \bu^n \right) , \bx^n \right) -
\sum_{\bu^n \in \cU^n} \sum_{\bx^n \in \cX^n}
D \epsilon \label{eq:app-4} \\
& = & \mbox{E}_{p_k} \left[ d_n \left( \pi^n \left( \bU^n \right) , \bX^n \right) \right]  -
\epsilon_0, \label{eq:app-5}
\end{eqnarray}
for all $\epsilon_0 > 0$, where \eqref{eq:app-4} follows since the distortion function is bounded by $D$. We now define
\begin{align}
\pi^n_{p_k} = arg\min_{\pi^n \, : \, \cU^n \rightarrow \cX^N} \mbox{E}_{p_k} \left[ d_n \left( \pi^n \left( \bU^n \right) , \bX^n \right) \right],
\label{eq:app-6}
\end{align}
which is well defined since there are finitely many codebook estimator functions. Then, we have
\begin{eqnarray}
\min_{\pi^n \, : \, \cU^n \rightarrow \cX^N} \mbox{E}_p \left[ d_n \left( \pi^n \left( \bU^n \right) , \bX^n \right) \right]
& = &  \mbox{E}_p \left[ d_n \left( \pi^n_{p} \left( \bU^n \right) , \bX^n \right) \right] \nonumber \\
& > &  \mbox{E}_{p_k} \left[ d_n \left( \pi^n_{p} \left( \bU^n \right) , \bX^n \right) \right]  -
\epsilon \label{eq:app-7} \\
& > & \mbox{E}_{p_k} \left[ d_n \left( \pi^n_{p_k} \left( \bU^n \right) , \bX^n \right) \right]  - \epsilon \label{eq:app-8}\\
& = & \min_{\pi^n \, : \, \cU^n \rightarrow \cX^N} \mbox{E}_{p_k} \left[ d_n \left( \pi^n \left( \bU^n \right) , \bX^n \right) \right] - \epsilon, \label{eq:app-9}
\end{eqnarray}
for all $\epsilon > 0$, where \eqref{eq:app-7} follows from \eqref{eq:app-5} and \eqref{eq:app-8} follows from the definition \eqref{eq:app-6}. Note that \eqref{eq:app-9} holds for any $n>N$, yielding
\begin{eqnarray}
\lim_{n \rightarrow \infty} \min_{\pi^n \, : \, \cU^n \rightarrow \cX^N}
\mbox{E}_p \left[ d \left( \pi^n \left( \bU^n \right) , \bX^n \right) \right]
&\geq &  \lim_{n \rightarrow \infty} \min_{\pi^n \, : \, \cU^n \rightarrow \cX^N}
\mbox{E}_{p_k} \left[ d \left( \pi^n \left( \bU^n \right) , \bX^n \right) \right] \nonumber \\
& \geq & \alpha, \label{eq:app-10}
\end{eqnarray}
where \eqref{eq:app-10} follows from \eqref{eq:app-3}. Hence, we proved that the limit lies in $P_{\alpha}$, i.e., $p \in P_{\alpha}$, and the set $P_{\alpha}$ is closed. Then, the proof follows from the Heine-Borel Theorem \cite{analysis}.
\qed

\bibliographystyle{IEEEbib}

\bibliography{msaf_references}

\end{document}